\documentclass[]{paper}

\title{An Explicit Solution for the Problem of Optimal Investment with Random Endowment}
\author{
  Michael Donisch\thanks{
    \myaffil\
    E-mail: michael.donisch@hotmail.de
  }
  \and
  Christoph Knochenhauer\thanks{
    \myaffil\
    E-mail: \myemail
  }
}
\date{}

\newcommand{\claims}{\mathcal{C}}

\begin{document}

%%%
%%% Title
%%%
\maketitle

%%%
%%% Abstract
%%%
\begin{abstract}
We consider the problem of optimal investment with random endowment in a Black--Scholes market for an agent with constant relative risk aversion.
Using duality arguments, we derive an explicit expression for the optimal trading strategy, which can be decomposed into the optimal strategy in the absence of a random endowment and an additive shift term whose magnitude depends linearly on the endowment-to-wealth ratio and exponentially on time to maturity.

%\msc\ ...
%
%\keywords\ ...
%
\end{abstract}

%%%
%%% Introduction
%%%
\section{Introduction}

In this brief note we consider the problem of optimal investment in the presence of a random endowment.
We consider a simple Black--Scholes market with an additional endowment process modeled as a geometric Brownian motion.
The objective is to optimize expected utility of terminal wealth for a power utility agent.

Optimal investment problems with random endowment have received significant interest over the last decades.
We refer to \cite{belak2022optimal} and the references therein for a recent in depth discussion and restrict our overview of the related literature to those works closest to ours.
In \cite{duffie1997hedging}, the authors consider a variant of our model over an infinite time horizon for an agent with a HARA utility function.
\cite{bick2013solving} compute almost optimal trading strategies in a very similar setting.
A dual approach was used in \cite{cvitanic2001utility} for a large class of utility functions and bounded endowment processes.
\cite{hu2005utility} and \cite{davis2006optimal} consider exponential utility functions, and \cite{horst2014forward} extend the results of \cite{hu2005utility} to the power utility case.
Finally, \cite{mostovyi2017optimal} and \cite{mostovyi2020optimal} provide results in very general market settings.
Most of the articles mentioned here aim for existence results for optimal trading strategies instead of explicit solutions.
This stems from the observation that optimal investment problems with random endowment are considered highly intractable---especially in incomplete market settings.

Closest to our paper is the recent article \cite{belak2022optimal}.
There, the authors consider a market model and optimization problem quite similar to ours, with the main difference being that we restrict to the case in which the noise driving the risky asset and the endowment process are perfectly correlated, making our market model complete.
The authors of \cite{belak2022optimal} follow a dynamic programming approach and develop a sophisticated argument to construct an optimal trading strategy by showing that the Hamilton--Jacobi--Bellman equation admits a classical solution coinciding with the value function, deriving a candidate optimal strategy in terms of the derivatives of the value function, and finally arguing that any optimal strategy must be equal to their candidate.
This approach comes with a significant technical burden, using deep results from viscosity solution theory and regularity results for non-linear partial differential equations, forcing the authors to work with restrictive assumptions on the set of admissible trading strategies.

In the present paper, we rely instead on duality arguments to solve the optimal investment problem.
As it turns out, this has significant advantages over the dynamic programming approach: we establish optimality of the trading strategy in a much larger class of strategies whilst relying on much more elementary arguments.
We regard this as one of the two main contributions of this article.

Our second main contribution stems from the fact that we obtain an explicit formula for the optimal trading strategy.
We show that the optimal fraction of wealth invested into the risky asset takes the simple form
\[
  \pi^*_t = \pi_M + \beta(t)\Bigl(\pi_M - \frac{\eta}{\sigma}\Bigr)\frac{E_t}{X^*_t},
  \qquad t\in[0,T],
\]
where $\pi_M$ is the optimal risky fraction in the absence of a random endowment, $\beta$ is a strictly positive and strictly decreasing function of time, $\eta$ is the volatility of the endowment process, $\sigma$ is the volatility of the risky asset, and $E/X^*$ is the endowment-to-wealth ratio.
We are not aware of any other setting in the literature in which an explicit optimal trading strategy in case of a random endowment has been found.

The remainder of this article is structured as follows.
In Section~\ref{sec:model-problem} we introduce the market model and formulate the optimal investment problem in the presence of a random endowment.
We proceed in Section~\ref{sec:duality} by introducing the dual problem and explicitly computing the optimal terminal wealth.
We conclude in Section~\ref{sec:optimal-strategy} by deriving an explicit expression for the optimal trading strategy and discussing its properties.

%%%
%%% Market Model and Problem Formulation
%%%
\section{Market Model and Problem Formulation}
\label{sec:model-problem}

We take as given a probability space $(\Om,\Sig,\Prob)$ sufficiently rich to support a standard one-dimensional Brownian motion $W = \{W_t\}_{t\in[0,T]}$ defined on a time interval $[0,T]$ with $T>0$.
Denote by $\Filt = \{\Filt_t\}_{t\in[0,T]}$ the filtration generated by $W$ augmented by the $\Prob$-nullsets.
We consider an agent investing in a Black--Scholes market consisting of a risk-free asset $S^0 = \{S^0_t\}_{t\in[0,T]}$ and a risky asset $S^1 = \{S^1_t\}_{t\in[0,T]}$ with dynamics
\[
  \de S^0_t = rS^0_t\de t\qquad\text{and}\qquad \de S^1_t = (r+\lambda)S^1_t\de t + \sigma S^1_t\de W_t,\qquad t\in[0,T].
\]
Here, $r\in\R$ denotes the risk-free rate, $\lambda\in\R$ the excess return, and $\sigma>0$ the volatility.
Without loss of generality, we let $S^0_0 = S^1_0 = 1$.
Next, we assume that the agent receives a random endowment $E = \{E_t\}_{t\in[0,T]}$ modeled by as a geometric Brownian motion, that is
\[
  \de E_t = \mu E_t\de t + \eta E_t \de W_t,\qquad t\in[0,T],
\]
where $\mu\in\R$ and $\eta>0$.
We suppose that $E_0>0$, which implies $E_t>0$ for all $t\in[0,T]$.

Trading strategies are modeled in terms of risky fractions, that is, $\R$-valued and $W$-integrable stochastic processes $\pi = \{\pi_t\}_{t\in[0,T]}$ representing the fraction of wealth invested into the risky asset $S^1$.
We denote the set of all such processes by $\Adm_0$.
Under the usual self-financing condition, given an initial wealth of $x>0$ and a risky fraction $\pi\in\Adm_0$, the agent's wealth $X^\pi = \{X^\pi_t\}_{t\in[0,T]}$ in the presence of the random endowment $E$ evolves as
\[
  \de X^\pi_t = \bigl((r + \lambda\pi_t)X^\pi_t + E_t\bigr)\de t + \sigma\pi_tX^\pi_t\de W_t,\qquad t\in[0,T],\qquad X^\pi_0 = x.
\]
Observe that since $x>0$ and $E_t>0$ for all $t\in[0,T]$, we must have $X^\pi_t>0$ for all $t\in[0,T]$.
We furthermore note for future reference that the dynamics of the wealth process can equivalently be expressed as
\[
  \de X^\pi_t = \bigl((r + \sigma\theta\pi_t)X^\pi_t + E_t\bigr)\de t + \sigma\pi_tX^\pi_t\de W_t,\qquad t\in[0,T],
\]
where $\theta\defined \lambda/\sigma$ is the market price of risk in the Black--Scholes model.
Recall that the market price of risk can be used to define the martingale density process $Z=\{Z_t\}_{t\in[0,T]}$ given by
\[
  \de Z_t = - \theta Z_t\de W_t,\qquad t\in[0,T],\qquad Z_0 = 1,
\]
which in turn gives rise to the risk neutral pricing measure $\ProbAlt\sim\Prob$ via the family of Radon--Nikodým densities
\[
  \frac{\de\ProbAlt}{\de\Prob}\Big|_{\Filt_t} \defined Z_t,\qquad t\in[0,T].
\]
Moreover, we denote by $H = \{H_t\}_{t\in[0,T]}$ the state-price deflator defined as $H \defined Z / S^0$ with dynamics
\[
  \de H_t = -rH_t\de t - \theta H_t\de W_t,\qquad t\in[0,T],\qquad H_0 = 1.
\]
Recall that any contingent claim with sufficiently integrable running payoff $C=\{C_t\}_{t\in[0,T]}$ and terminal payoff $\xi$ has a unique arbitrage-free price given by
\[
  \E^\ProbAlt\Bigl[\int_0^T e^{-rt}C_t\de t + e^{-rT}\xi\Bigr] = \E\Bigl[\int_0^T C_tH_t\de t + \xi H_T\Bigr],
\]
signifying the role of the state-price deflator.

Let us now turn towards the formulation of the optimal investment problem.
We begin by fixing a power utility function
\[
  U : (0,\infty) \to \R,\qquad x \mapsto U(x) \defined \frac{x^{1-\gamma}}{1-\gamma},
\]
where $\gamma\in(0,1)\cup(1,\infty)$ denotes the agent's constant relative risk aversion parameter.
Next, let us introduce the set of admissible risky fractions as
\[
  \Adm \defined \bigl\{ \pi\in\Adm_0 : \E\bigl[U_-(X^\pi_T)\bigr] < \infty \bigr\},
\]
where $U_- \defined \max\{-U,0\}$ denotes the negative part of $U$.
With this, the problem of optimal investment in a Black--Scholes market with random endowment can be formulated as
\begin{equation}\label{eq:op}\tag{OPT}
  \text{maximize }\E\bigl[ U\bigl(X^\pi_T\bigr) \bigr]\text{ over all }\pi\in\Adm.
\end{equation}

%%%
%%% The Dual Problem and Optimal Terminal Wealth
%%%
\section{The Dual Problem and Optimal Terminal Wealth}
\label{sec:duality}

We follow a dual approach to solve the optimization problem \eqref{eq:op}.
We start with several auxiliary results facilitating the formulation of the dual problem.
Then we use a Lagrangian approach to compute the optimal terminal wealth in the presence of a random endowment and finally conclude this section by constructing a risky fraction which replicates the optimal terminal wealth.

\subsection{Preliminaries: Budget Constraint and Replication}

In this subsection, we gather several results which play a crucial role throughout the entire article.
The first result explicitly computes the arbitrage-free price of the random endowment.

\begin{lemma}[Arbitrage-Free Price of Endowment]
For any $t\in[0,T]$, it holds that
\[
  P_t \defined \E\Bigl[\int_0^t E_s H_s \de s\Bigr] = E_0\frac{\exp\{(\mu-r-\eta\theta)t\}-1}{\mu-r-\eta\theta} < \infty.\closeEqn
\]
\end{lemma}

\begin{proof}
We first make use of the fact that the stochastic differential equations for $E$ and $H$ have explicit solutions given by
\begin{equation}\label{eq:explicit-solutions}
  E_s = E_0e^{(\mu - \eta^2/2)s + \eta W_s}
  \qquad\text{and}\qquad
  H_s = e^{-( r + \theta^2/2)s - \theta W_s},
  \qquad s\in[0,T],
\end{equation}
from which we conclude that
\[
  \E\bigl[E_s H_s\bigr]
  = \E\Bigl[E_0e^{(\mu - \eta^2/2)s + \eta W_s}e^{-(r + \theta^2/2)s - \theta W_s}\Bigr]
  = E_0e^{(\mu - r - \eta^2/2 - \theta^2/2)s}\E\Bigl[e^{(\eta-\theta)W_s}\Bigr].
\]
The expectation can be identified as the moment generating function of the normal distribution, implying that
\begin{equation}\label{eq:EH-expectation}
  \E\bigl[E_s H_s\bigr]
  = E_0e^{(\mu - r - \eta^2/2 - \theta^2/2)s}e^{s(\eta-\theta)^2/2}
  = E_0e^{(\mu - r - \eta\theta)s}.
\end{equation}
But then Tonelli's theorem yields
\[
  \E\Bigl[\int_0^t E_s H_s \de s\Bigr]
  = \int_0^t \E\bigl[E_s H_s\bigr] \de s
  = E_0\frac{\exp\{(\mu-r-\eta\theta)t\}-1}{\mu-r-\eta\theta},
\]
thus concluding the proof.
\end{proof}

Let us take special note of the constant
\begin{equation}\label{eq:endow-price}
  P_T = \E\Bigl[\int_0^T E_t H_t \de t\Bigr] = E_0\frac{\exp\{(\mu-r-\eta\theta)T\}-1}{\mu-r-\eta\theta},
\end{equation}
which corresponds to the arbitrage-free price of the random endowment over the entire trading period.
This constant will appear in numerous equations below.

Moving on, we establish a variant of the classical budget constraint lemma in the presence of a random endowment.

\begin{proposition}[Budget Constraint]\label{prop:budget-constraint}
Let $\pi\in\Adm_0$ be any risky fraction. Then the process $Y^\pi = \{Y^\pi_t\}_{t\in[0,T]}$ given by
\[
  Y^\pi_t \defined X^\pi_t H_t - \int_0^t E_s H_s\de s,\qquad t\in[0,T],
\]
is a supermartingale under $\Prob$.
In particular, the budget constraint
\[
  \E\Bigl[X^\pi_t H_t - \int_0^t E_sH_s\de s\Bigr] = \E[Y^\pi_t] \leq Y^\pi_0 = x
\]
holds for all $t\in[0,T]$.\close
\end{proposition}

\begin{proof}
It\={o}'s product rule reveals
\begin{equation}\label{eq:proof-budget}
  \de Y^\pi_t
  = X^\pi_t\de H_t + H_t\de X^\pi_t + \de X^\pi_t\de H_t - E_tH_t\de t\\
  = (\sigma\pi_t-\theta)X^\pi_tH_t\de W_t,\quad t\in[0,T],
\end{equation}
from which we conclude that $Y^\pi$ is a (continuous) local martingale.
Now let $\{\tau_k\}_{k\in\N}$ be a localizing sequence and fix $s,t\in[0,T]$ with $s<t$.
Then
\begin{align*}
  Y^\pi_s = \liminf_{k\to\infty} Y^\pi_{s\wedge\tau_k}
  &=\liminf_{k\to\infty}\E\bigl[Y^\pi_{t\wedge\tau_k}\givenbig\Filt_s]\\
  &\geq \liminf_{k\to\infty}\E\bigl[ X^\pi_{t\wedge\tau_k}H_{t\wedge\tau_k} \givenbig\Filt_s\bigr] - \liminf_{k\to\infty}\E\Bigl[ \int_0^{t\wedge\tau_k} E_rH_r \de r \givenBig\Filt_s\Bigr].
\end{align*}
Since $X^\pi,H,E > 0$, we can apply Fatou's lemma for the first expectation and monotone convergence for the second expectation to conclude that
\[
  Y^\pi_s
  \geq \E\Bigl[ X^\pi_tH_t - \int_0^t E_rH_r \de r \givenBig\Filt_s\Bigr]
  = \E[Y^\pi_t\given\Filt_s],
\]
so $Y^\pi$ is indeed a supermartingale.
\end{proof}

The final preliminary result concerns replication of arbitrary contingent claims in the presence of random endowments.

\begin{proposition}[Replication]\label{prop:replication}
Let $\xi\geq 0$ be an $\Filt_T$-measurable random variable such that
\[
  x_\xi \defined \E\Bigl[\xi H_T - \int_0^T E_tH_t\de t\Bigr] < \infty.
\]
Then there exists a risky fraction $\pi\in\Adm_0$ such that
\[
  X^\pi_0 = x_\xi\qquad\text{and}\qquad X^\pi_T = \xi.\closeEqn
\]
\end{proposition}

\begin{proof}
By martingale representation, there exists a $W$-integrable process $K$ such that
\[
  \xi H_T - \int_0^T E_tH_t\de t = x_\xi + \int_0^T K_t\de W_t
  \qquad\text{and}\qquad
  M\defined \int_0^{\argdot} K_t\de W_t\text{ is a martingale.}
\]
Now define a process $X = \{X_t\}_{t\in[0,T]}$ by
\begin{equation}\label{eq:proof-repl}
  X_t \defined \frac{1}{H_t}\Bigl[x_\xi + M_t + \int_0^t E_sH_s\de s\Bigr],\qquad t\in[0,T],
\end{equation}
and observe that $X_0 = x_\xi$ and
\[
  X_T = \frac{1}{H_T}\Bigl[x_\xi + M_T + \int_0^T E_sH_s\de s\Bigr] = \frac{1}{H_T}\Bigl[\xi H_T - \int_0^T E_tH_t\de t + \int_0^T E_sH_s\de s\Bigr] = \xi.
\]
It remains to show that there exists $\pi\in\Adm_0$ such that $X^\pi = X$.
For this, we first set
\begin{equation}\label{eq:feedback-repl}
  \hat{\pi}(t,\hat{x}) \defined \frac{1}{\sigma}\Bigl(\frac{K_t}{\hat{x}H_t}\ind{\{\hat{x}\neq 0\}} + \theta\Bigr),\qquad (t,\hat{x})\in[0,T]\times\R,
\end{equation}
and consider the stochastic differential equation
\[
  \de\hat{X}_t = \bigl((r + \sigma\theta\hat{\pi}(t,\hat{X}_t))\hat{X}_t + E_t\bigr)\de t + \sigma\hat{\pi}(t,\hat{X}_t)\hat{X}_t\de W_t,\qquad t\in[0,T],\qquad \hat{X}_0 = x_\xi.
\]
Using the definition of $\hat{\pi}$, the dynamics can in fact be written as
\[
  \de\hat{X}_t = \Bigl(\theta\frac{K_t}{H_t} + E_t + \bigl(r+\theta^2\bigr)\hat{X}_t\Bigr)\de t + \Bigl(\frac{K_t}{H_t} + \theta\hat{X}_t\Bigr)\de W_t,\qquad t\in[0,T],
\]
from which we see that a unique strong solution $\hat{X}$ indeed exists.
Moreover, upon defining
\begin{equation}\label{eq:repl-from-feedback}
  \pi_t \defined \hat{\pi}\bigl(t,\hat{X}_t\bigr),\qquad t\in[0,T],
\end{equation}
it follows that
\[
  \hat{X} = X^\pi\qquad\text{and hence}\qquad \pi = \hat{\pi}\bigl(\argdot,X^\pi\bigr).
\]
Now note that by rearranging Equation \eqref{eq:proof-repl} and using the definition of $M$, we can write
\[
  X_tH_t - \int_0^t E_sH_s\de s = x_\xi + M_t = x_\xi + \int_0^t K_s\de W_s,\qquad t\in[0,T].
\]
Conversely, as in the proof of Proposition~\ref{prop:budget-constraint} or more precisely Equation \eqref{eq:proof-budget}, we have
\[
  X^\pi_tH_t - \int_0^t E_sH_s\de s = x_\xi + \int_0^t (\sigma\pi_s-\theta)X^\pi_sH_s \de W_s,\qquad t\in[0,T].
\]
Thus, using the definition of $\pi$,
\[
  \bigl(X_t - X^\pi_t\bigr)H_t = \int_0^t \bigl(K_s - (\sigma\pi_s-\theta)X^\pi_sH_s\bigr) \de W_s = 0,\qquad t\in[0,T].
\]
As this holds for all $t\in[0,T]$ and $H>0$, we conclude that $X = X^\pi$.
\end{proof}

\subsection{Optimal Terminal Wealth with Random Endowment}

We now turn towards solving the optimal investment problem with random endowment \eqref{eq:op}.
The first step is to compute the optimal terminal wealth, for which we employ a Lagrangian approach.
For this, we first introduce the set $\claims(x)$ of $\Filt_T$-measurable random variables $\xi\geq 0$ satisfying
\[
  \E\Bigl[\xi H_T - \int_0^T E_tH_t\de t\Bigr] \leq x,
\]
where we recall that $x$ denotes the agent's initial wealth.
According to Proposition~\ref{prop:budget-constraint} (Budget Constraint), the set $\claims(x)$ can be interpreted as the set of option payoffs which can be replicated with an initial wealth of $x$.
To wit, we expect that
\[
  \sup_{\pi\in\Adm}\E\bigl[U\bigl(X^\pi_T\bigr)\bigr] = \sup_{\xi\in\claims(x)}\E\bigl[U(\xi)\bigr],
\]
suggesting that a first step towards solving the investment problem is to solve the static optimization problem on the right-hand side.

Towards this, denote by $\lambda > 0$ a Lagrange multiplier, fix $\xi\in\claims(x)$, and note that by definition of $\claims(x)$ it holds that
\begin{align}
  \E\bigl[U(\xi)\bigr]
  &\leq \E\bigl[U(\xi)\bigr] + \lambda\Bigl(x - \E\Bigl[\xi H_T - \int_0^T E_tH_t\de t\Bigr]\Bigr)\notag\\
  &= \E\bigl[U(\xi) - \lambda\xi H_T\bigr] + \lambda(x+P_T),\label{eq:dual}
\end{align}
where $P_T$ was defined in Equation \eqref{eq:endow-price}.
Now write $\tilde{U} : (0,\infty) \to \R$ for the Legendre--Fenchel transform of $U$, that is
\[
  \tilde{U}(y) \defined \sup_{x>0}\bigl\{U(x) - yx\bigr\},\qquad y > 0.
\]
Using this in Equation \eqref{eq:dual} and the fact that $\xi\in\claims(x)$ and $\lambda>0$ are chosen arbitrarily implies
\begin{equation}\label{eq:duality}\tag{DUAL}
  \sup_{\xi\in\claims(x)}\E\bigl[U(\xi)\bigr] \leq \inf_{\lambda>0}\Bigl\{\E\bigl[\tilde U\bigl(\lambda H_T\bigr)\bigr] + \lambda(x+P_T)\Bigr\}.
\end{equation}
The right-hand side is the referred to as the dual problem of \eqref{eq:op}.
From the derivation of the dual problem, it is straightforward to see under which conditions we have equality in \eqref{eq:duality}:
\begin{enumerate}
\item[(1)] There needs to exist
\[
  \lambda^*\in\arg\min_{\lambda>0}\Bigl\{\E\bigl[\tilde U\bigl(\lambda H_T\bigr)\bigr] + \lambda(x+P_T)\Bigr\}.
\]
\item[(2)] There needs to exist $\xi^*\in\claims(x)$ such that
\[
  \tilde U\bigl(\lambda^* H_T\bigr) = U(\xi^*) - \lambda^*\xi^* H_T.
\]
\item[(3)] The claim $\xi^*$ needs to satisfy the budget constraint with equality, that is
\[
  \E\Bigl[\xi^* H_T - \int_0^T E_tH_t\de t\Bigr] = x,
  \qquad\text{that is,}\qquad
  \E\bigl[\xi^* H_T\bigr] = x + P_T.
\]
\end{enumerate}
We proceed to address these requirements one at a time to derive a candidate solution $(\lambda^*,\xi^*)$.
We begin by taking a closer look at the Legendre--Fenchel transform $\tilde{U}$.
As $U$ is smooth, it is straightforward to compute the optimizer in the definition of $\tilde{U}$.
Indeed, a quick calculation shows that
\[
  \tilde{U}(y) = \sup_{x>0}\bigl\{U(x) - yx\bigr\} = U\bigl(I(y)\bigr) - yI(y),\qquad y>0,
\]
where
\[
  I : (0,\infty) \to \R,\qquad y \mapsto I(y) \defined (U')^{-1}(y) = y^{-1/\gamma},
\]
denotes the inverse marginal utility of $U$.
To obtain a candidate minimizer $\lambda^*$ for the dual problem, we proceed heuristically by noting that the first-order condition reads
\begin{align*}
  0 &= \frac{\de}{\de\lambda}\Bigl[\E\bigl[\tilde U\bigl(\lambda H_T\bigr)\bigr] + \lambda(x+P_T)\Bigr]_{\lambda = \lambda^*}\\
  &= \frac{\de}{\de\lambda}\Bigl[\E\bigl[U\bigl(I(\lambda H_T)\bigr) - \lambda H_TI(\lambda H_T)\bigr] + \lambda(x+P_T)\Bigr]_{\lambda = \lambda^*}\\
  &= - \E\bigl[H_TI(\lambda^*H_T)\bigr] + x + P_T.
\end{align*}
Defining the function
\[
  \chi : (0,\infty) \to (0,\infty),\qquad \lambda \mapsto \chi(\lambda) \defined \E\bigl[H_TI(\lambda H_T)\bigr] = \lambda^{-1/\gamma}\E\bigl[H_T^{-(1-\gamma)/\gamma}\bigr],
\]
we conclude that we need to find $\lambda^*$ such that $\chi(\lambda^*) = x + P_T$, which is a straightforward exercise in the particular case considered here.

\begin{lemma}[Candidate Optimal Lagrange Multiplier $\lambda^*$]
There exists a unique number $\lambda^*>0$ such that $\chi(\lambda^*) = x + P_T$, which is explicitly given by
\[
  \lambda^* = \exp\Bigl\{(1-\gamma)\Bigl(r + \frac{1}{\gamma}\frac{\theta^2}{2}\Bigr)T\Bigr\}(x + P_T)^{-\gamma}.\closeEqn
\]
\end{lemma}

\begin{proof}
It is obvious from the definition of $\chi$ that is is invertible and
\[
  \lambda^* = \chi^{-1}(x + P_T) = (x + P_T)^{-\gamma}\E\bigl[H_T^{-(1-\gamma)/\gamma}\bigr]^\gamma,
\]
so it remains to compute the expectation on the right-hand side.
Using the explicit formula we have for $H_T$ given in Equation \eqref{eq:explicit-solutions}, we have
\begin{equation}\label{eq:H-power}
  \E\bigl[H_T^{-(1-\gamma)/\gamma}\bigr]
  = \exp\Bigl\{\frac{1-\gamma}{\gamma}\Bigl(r + \frac{\theta^2}{2}\Bigr)T\Bigr\}\E\Bigl[e^{\frac{(1-\gamma)\theta}{\gamma}W_T}\Bigr]
  = \exp\Bigl\{\frac{1-\gamma}{\gamma}\Bigl(r + \frac{1}{\gamma}\frac{\theta^2}{2}\Bigr)T\Bigr\}.
\end{equation}
Plugging this into the formula for $\lambda^*$ above concludes the proof.
\end{proof}

Having computed the candidate optimal Lagrange multiplier $\lambda^*$, we can proceed with the second step in our program and derive the candidate optimal terminal wealth $\xi^*\in\claims(x)$, which we recall has to satisfy
\[
  \tilde{U}(\lambda^*H_T) = U(\xi^*) - \lambda^*\xi^*H_T.
\]
Put differently, $\xi^*$ has to maximize $\tilde{U}(\lambda^*H_T)$.
The third requirement furthermore dictates that $\xi^*$ has to satisfy the budget constraint with equality, that is
\[
  \E\bigl[\xi^*H_T\bigr] = x + P_T,
\]
which implies, in particular, $\xi^*\in\claims(x)$.

\begin{lemma}[Candidate Optimal Terminal Wealth and Budget Constraint]\label{lem:candidate-optimal-wealth}
The random variable
\[
  \xi^* \defined I\bigl(\lambda^*H_T\bigr) = \exp\Bigl\{-\frac{1-\gamma}{\gamma}\Bigl(r + \frac{1}{\gamma}\frac{\theta^2}{2}\Bigr)T\Bigr\}(x + P_T)  H_T^{-1/\gamma} > 0
\]
satisfies
\[
  \tilde{U}(\lambda^*H_T) = U(\xi^*) - \lambda^*\xi^*H_T
  \qquad\text{and}\qquad
  \E\bigl[\xi^*H_T\bigr] = x + P_T.
\]
In particular, $\xi^*\in\claims(x)$.\close
\end{lemma}

\begin{proof}
Recall that
\[
  \tilde{U}(y) = U\bigl(I(y)\bigr) - yI(y),\qquad y>0,
\]
from which we immediately see that $\xi^*$ satisfies $\tilde{U}(\lambda^*H_T) = U(\xi^*) - \lambda^*\xi^*H_T$.
Moreover, using the explicit formula for $\xi^*$ we see that
\[
  \E\bigl[\xi^*H_T\bigr]
  = (x + P_T)\exp\Bigl\{-\frac{1-\gamma}{\gamma}\Bigl(r + \frac{1}{\gamma}\frac{\theta^2}{2}\Bigr)T\Bigr\}\E\bigl[H_T^{-(1-\gamma)/\gamma}\bigr]
  = x + P_T,
\]
where the second identity is a direct consequence of Equation~\eqref{eq:H-power}.
\end{proof}

Let us highlight here that, for now, $\lambda^*$ and $\xi^*$ are only candidate optimizers and we still have to present a rigorous argument showing that $\xi^*$ is indeed the optimal terminal wealth.
This is achieved by the following theorem.

\begin{theorem}[Optimal Terminal Wealth]\label{thm:optimal-terminal-wealth}
The optimal terminal wealth for the optimal investment problem with random endowment specified in \eqref{eq:op} is given by
\[
  \xi^* = \exp\Bigl\{-\frac{1-\gamma}{\gamma}\Bigl(r + \frac{1}{\gamma}\frac{\theta^2}{2}\Bigr)T\Bigr\}(x + P_T)  H_T^{-1/\gamma}.
\]
More precisely, there exists a risky fraction process $\pi^*\in\Adm$ with corresponding wealth process $X^*\defined X^{\pi^*}$ such that
\[
  X^*_0 = x,\qquad X^*_T = \xi^*,\qquad\text{and}\qquad \sup_{\pi\in\Adm}\E\bigl[U\bigl(X^\pi_T\bigr)\bigr] = \E\bigl[U(\xi^*)\bigr].
\]
In particular, $\pi^*$ is optimal.\close
\end{theorem}

\begin{proof}
By Lemma~\ref{lem:candidate-optimal-wealth} (Candidate Optimal Terminal Wealth and Budget Constraint) we have
\[
  \E\Bigl[\xi^*H_T - \int_0^T E_tH_t\de t\Bigr] = x + P_T - P_T = x,
\]
so that by Proposition~\ref{prop:replication} (Replication) there exists $\pi^*\in\Adm_0$ such that $X^*_0 = x$ and $X^*_T = \xi^*$, where we write $X^* \defined X^{\pi^*}$ as short-hand notation.
It remains to verify that we even have $\pi^*\in\Adm$ and $\E[U(X^\pi_T)] \leq \E[U(\xi^*)]$ for any other risky fraction $\pi\in\Adm$.
For this, we first note that concavity of $U$ implies
\begin{equation}\label{eq:concavity}
  U(\hat{y}) \geq U(y) + U'(\hat{y})(\hat{y}-y),\qquad y,\hat{y}>0.
\end{equation}
Applying this inequality with $\hat{y} \defined I(\lambda^*H_T) = \xi^* = X^*_T$ and $y = 1$ yields
\[
  U(X^*_T)
  \geq U(1) + \lambda^*H_T\bigl(I(\lambda^*H_T)-1\bigr)
  = U(1) + (\lambda^*H_T)^{-(1-\gamma)/\gamma} - \lambda^*H_T.
\]
In particular, the negative part of the left-hand side must be dominated by the negative part of the right-hand side, which can be further estimated from above by its absolute value.
Hence
\[
  \E\bigl[U_-(X^*_T)\bigr]
  \leq |U(1)| + (\lambda^*)^{-(1-\gamma)/\gamma}\E\Bigl[H_T^{-(1-\gamma)/\gamma}\Bigr] + \lambda^*\E[H_T] < \infty,
\]
which is to say that $\pi^*\in\Adm$.
Finally, regarding optimality of $\pi^*$, let us fix another risky fraction $\pi\in\Adm$ and apply the inequality in Equation~\eqref{eq:concavity} with $\hat{y} = X^*_T = \xi^* = I(\lambda^*H_T)$ and $y = X^\pi_T$, showing that
\begin{align*}
  \E\bigl[U(X^*_T)\bigr]
  &\geq \E\bigl[U\bigl(X^\pi_T\bigr)\bigr] + \lambda^*\Bigl(\E\bigl[X^*_TH_T\bigr] - \E\bigl[X^\pi_TH_T\bigr]\Bigr).
\end{align*}
Now $\E[X^*_TH_T] = \E[\xi^*H_T] = x + P_T$ by Lemma~\ref{lem:candidate-optimal-wealth} (Candidate Optimal Terminal Wealth and Budget Constraint), whereas $\E[X^\pi_TH_T] \geq x + P_T$ by Proposition~\ref{prop:budget-constraint} (Budget Constraint).
This shows that
\[
  \E\bigl[U(X^*_T)\bigr] \geq \E\bigl[U\bigl(X^\pi_T\bigr)\bigr] + \lambda^*\bigl(x + P_T - (x + P_T)\bigr) = \E\bigl[U\bigl(X^\pi_T\bigr)\bigr],
\]
establishing the optimality of $\pi^*$.
\end{proof}

%%%
%%% The Optimal Strategy
%%%
\section{The Optimal Strategy}
\label{sec:optimal-strategy}

With the optimal terminal wealth given explicitly and the existence of an optimal risky fraction $\pi^*$ at hand, we now proceed to derive an explicit formula for $\pi^*$.
The starting point is Proposition~\ref{prop:replication} (Replication), whose proof contains a method to construct $\pi^*$ as the replication strategy of the optimal terminal wealth $\xi^*$.
Taking a closer look at the proof---in particular Equations~\eqref{eq:feedback-repl} and \eqref{eq:repl-from-feedback}---we see that
\[
  \pi^*_t = \frac{1}{\sigma}\Bigl(\frac{K^*_t}{X^*_tH_t} + \theta\Bigr),\qquad t\in[0,T],
\]
where $X^*_t = X^{\pi^*}_t>0$ and $K^*$ is the $W$-integrable process satisfying
\[
  x + \int_0^T K^*_t \de W_t = \xi^*H_T - \int_0^T E_tH_t\de t.
\]
Since it was argued that the stochastic integral on the left-hand side is a martingale, we must in fact even have
\begin{equation}\label{eq:martingale-representation}
  x + \int_0^t K^*_s\de W_s = \E\Bigl[\xi^*H_T - \int_0^T E_sH_s\de s\givenBig\Filt_t\Bigr],\qquad t\in[0,T].
\end{equation}
So if we are able to compute the right-hand side more explicitly, we obtain a more explicit expression for $K^*$ and therefore also $\pi^*$.

\begin{lemma}[Explicit Martingale Representation]\label{lem:martingale-representation}
We have
\[
  \E\Bigl[\xi^*H_T - \int_0^T E_sH_s\de s\givenBig\Filt_t\Bigr] = \alpha(t)H_t^{-(1-\gamma)/\gamma} - \int_0^t E_sH_s\de s - \beta(t)E_tH_t,\quad t\in[0,T],
\]
where
\begin{align*}
  \alpha(t) &\defined (x + P_T)\exp\Bigl\{-\frac{1-\gamma}{\gamma}\Bigl(r + \frac{1}{\gamma}\frac{\theta^2}{2}\Bigr)t\Bigr\}, & t&\in[0,T],\\
  \beta(t) &\defined \frac{\exp\bigl\{\bigl(\mu - r - \eta\theta\bigr)(T-t)\bigr\} - 1}{\mu - r - \eta\theta}, & t&\in[0,T].
\end{align*}
As a consequence, it holds that
\[
  K^*_t = (\theta-\eta)\beta(t)E_t H_t + \frac{1-\gamma}{\gamma}\theta\alpha(t)H_t^{-(1-\gamma)/\gamma},\qquad t\in[0,T].\closeEqn
\]
\end{lemma}

\begin{proof}
The proof proceeds in several steps.
Throughout the entire proof, we fix $t\in[0,T]$.

Step 1: Computation of $\E[\xi^*H_T\given\Filt_t]$.
We use the explicit formula for $\xi^*$ from Theorem~\ref{thm:optimal-terminal-wealth} (Optimal Terminal Wealth) to obtain
\[
  \E\bigl[\xi^*H_T\givenbig\Filt_t\bigr]
  =  (x + P_T)\exp\Bigl\{-\frac{1-\gamma}{\gamma}\Bigl(r + \frac{1}{\gamma}\frac{\theta^2}{2}\Bigr)T\Bigr\} H_t^{-(1-\gamma)/\gamma} \E\Bigl[\Bigl(\frac{H_T}{H_t}\Bigr)^{-(1-\gamma)/\gamma}\givenBig\Filt_t\Bigr].
\]
Using the explicit formula for $H$ in Equation~\eqref{eq:explicit-solutions}, it follows moreover that
\begin{align*}
  \E\Bigl[\Bigl(\frac{H_T}{H_t}\Bigr)^{-(1-\gamma)/\gamma}\givenBig\Filt_t\Bigr]
  &= \exp\Bigl\{\frac{1-\gamma}{\gamma}\Bigl(r + \frac{\theta^2}{2}\Bigr)(T-t)\Bigr\}\E\Bigl[\exp\Bigl\{-\frac{(1-\gamma)\theta}{\gamma}(W_T-W_t)\Bigr\}\Bigr]\\
  &= \exp\Bigl\{\frac{1-\gamma}{\gamma}\Bigl(r + \frac{1}{\gamma}\frac{\theta^2}{2}\Bigr)(T-t)\Bigr\}.
\end{align*}
In combination, this shows that
\[
  \E\bigl[\xi^*H_T\givenbig\Filt_t\bigr]
  = (x + P_T)\exp\Bigl\{-\frac{1-\gamma}{\gamma}\Bigl(r + \frac{1}{\gamma}\frac{\theta^2}{2}\Bigr)t\Bigr\}H_t^{-(1-\gamma)/\gamma}
  = \alpha(t)H_t^{-(1-\gamma)/\gamma}.
\]

Step 2: Computation of $\E[\int_0^T E_sH_s\de s\given\Filt_t]$.
We decompose the conditional expectation by writing
\[
  \E\bigl[\int_0^T E_sH_s\de s\givenbig\Filt_t\bigr]
  = \int_0^t E_sH_s\de s + E_tH_t\int_t^T \E\Bigl[\frac{E_sH_s}{E_tH_t}\givenBig\Filt_t\Bigr]\de s.
\]
Using once again the explicit formulas for $E$ and $H$ in Equation~\eqref{eq:explicit-solutions} and arguing as in Equation~\eqref{eq:EH-expectation}, it follows that for every $s\in[t,T]$ we have
\begin{align*}
  \E\Bigl[\frac{E_sH_s}{E_tH_t}\givenBig\Filt_t\Bigr]
  &= \exp\Bigl\{\Bigl(\mu - r - \frac{1}{2}\eta^2 - \frac{1}{2}\theta^2\Bigr)(s-t)\Bigr\}\E\Bigl[\exp\Bigl\{(\eta - \theta) (W_s-W_t)\Bigr\}\Bigr]\\
  &= \exp\bigl\{\bigl(\mu - r - \eta\theta\bigr)(s-t)\bigr\}.
\end{align*}
We conclude that
\begin{align*}
  \E\bigl[\int_0^T E_sH_s\de s\givenbig\Filt_t\bigr]
  &= \int_0^t E_sH_s\de s + E_tH_t \int_t^T \exp\bigl\{\bigl(\mu - r - \eta\theta\bigr)(s-t)\bigr\} \de s\\
  &= \int_0^t E_sH_s\de s + \frac{\exp\bigl\{\bigl(\mu - r - \eta\theta\bigr)(T-t)\bigr\} - 1}{\mu - r - \eta\theta}E_tH_t\\
  &= \int_0^t E_sH_s\de s + \beta(t)E_tH_t.
\end{align*}
Combining this with the calculations performed in the first step of this proof, we arrive at
\[
  \E\Bigl[\xi^*H_T - \int_0^T E_sH_s\de s\givenBig\Filt_t\Bigr]
  = \alpha(t)H_t^{-(1-\gamma)/\gamma} - \int_0^t E_sH_s\de s - \beta(t)E_tH_t
\]
as claimed.

Step 3: Determining $K^*$.
Consider the process $M^* = \{M^*_t\}_{t\in[0,T]}$ defined by
\[
  M^*_t \defined \alpha(t)H_t^{-(1-\gamma)/\gamma} - \int_0^t E_sH_s\de s - \beta(t)E_tH_t.
\]
By It\={o}'s formula, we have
\begin{align*}
  \de M^*_t
  &= -\frac{1-\gamma}{\gamma}\alpha(t)H_t^{-1/\gamma}\de H_t + \frac{1}{2}\frac{1-\gamma}{\gamma^2}\alpha(t)H_t^{-(1+\gamma)/\gamma}\de H_t\de H_t + H_t^{-(1-\gamma)/\gamma}\de\alpha(t)\\
  &\hspace{2.5cm} - E_tH_t\de t - \beta(t)E_t\de H_t - \beta(t)H_t\de E_t - \beta(t)\de E_t\de H_t - E_tH_t\de\beta(t)\\
  &= \Bigl[(\theta-\eta)\beta(t)E_t H_t + \frac{1-\gamma}{\gamma}\theta\alpha(t)H_t^{-(1-\gamma)/\gamma}\Bigr]\de W_t.
\end{align*}
But by Equation~\eqref{eq:martingale-representation}, we know that
\[
  K^*_t \de W_t = \de M^*_t = \Bigl[(\theta-\eta)\beta(t)E_t H_t + \frac{1-\gamma}{\gamma}\theta\alpha(t)H_t^{-(1-\gamma)/\gamma}\Bigr]\de W_t,
\]
so the result follows by comparing the integrands.
\end{proof}

The previous lemma implies that the optimal risky fraction can be written as
\[
  \pi^*_t = \frac{1}{\sigma}\Bigl(\frac{K^*_t}{X^*_tH_t} + \theta\Bigr),\qquad t\in[0,T],
\]
where
\[
  K^*_t = (\theta-\eta)\beta(t)E_t H_t + \frac{1-\gamma}{\gamma}\theta\alpha(t)H_t^{-(1-\gamma)/\gamma},\qquad t\in[0,T].
\]
It turns out, however, that this expression can be simplified even further, leading to the following main result of this section.

\begin{theorem}[Explicit Optimal Risky Fraction]
The optimal risky fraction $\pi^*\in\Adm$ can be represented as
\[
  \pi^*_t = \pi_M + \beta(t)\Bigl(\pi_M - \frac{\eta}{\sigma}\Bigr)\frac{E_t}{X^*_t},
  \qquad t\in[0,T],
\]
where $\pi_M \defined \theta/(\gamma\sigma)$ is the optimal risky fraction in the absence of a random endowment.\close
\end{theorem}

\begin{proof}
Fix $t\in[0,T]$.
From the proof of Proposition~\ref{prop:replication} (Replication), the optimal wealth process $X^*$ can be written as
\[
  X^*_t = \frac{1}{H_t}\Bigl(x + \int_0^t K^*_s\de W_s + \int_0^t E_sH_s\de s\Bigr);
\]
see Equation~\eqref{eq:proof-repl}.
Using Equation~\eqref{eq:martingale-representation} and Lemma~\ref{lem:martingale-representation} (Explicit Martingale Representation), we conclude that
\begin{align}
  X^*_t
  &= \frac{1}{H_t}\Bigl(\alpha(t)H_t^{-(1-\gamma)/\gamma} - \int_0^t E_sH_s\de s - \beta(t)E_tH_t + \int_0^t E_sH_s\de s\Bigr)\notag\\
  &= \alpha(t)H_t^{-1/\gamma} - \beta(t)E_t.\label{eq:optimal-wealth-process}
\end{align}
Plugging this and the expression for $K^*$ into the formula for the optimal risky fraction yields
\begin{align*}
  \pi^*_t
  = \frac{1}{\sigma}\Bigl(\frac{K^*_t}{X^*_tH_t} + \theta\Bigr)
  &= \frac{1}{\sigma}\Bigl(\frac{(\theta-\eta)\beta(t)E_t H_t + \frac{1-\gamma}{\gamma}\theta\alpha(t)H_t^{-(1-\gamma)/\gamma}}{\alpha(t)H_t^{-(1-\gamma)/\gamma} - \beta(t)E_tH_t} + \theta\Bigr)\\
  &= \frac{\theta}{\gamma\sigma} + \beta(t)\Bigl(\frac{\theta}{\gamma\sigma} - \frac{\eta}{\sigma}\Bigr)\frac{E_t}{\alpha(t)H_t^{-1/\gamma} - \beta(t)E_t}.
\end{align*}
Using the definition of $\pi_M$ and the expression for $X^*$ in Equation~\eqref{eq:optimal-wealth-process} yields the result.
\end{proof}

Summing up the results of this section, we can conclude that the optimal risky fraction $\pi^*$ takes the form
\[
  \pi^*_t = \pi_M + \beta(t)\Bigl(\pi_M - \frac{\eta}{\sigma}\Bigr)\frac{E_t}{X^*_t},
  \qquad t\in[0,T],
\]
where $\pi_M$ is the optimal risky fraction in the absence of a random endowment and
\[
  \beta(t) \defined \frac{\exp\bigl\{(\mu - r - \eta\theta)(T-t)\bigr\} - 1}{\mu - r - \eta\theta},\qquad t\in[0,T].
\]
From the expression for $\pi^*$, we see that the presence of a random endowment causes an additive shift in the optimal risky fraction by an amount of
\begin{equation}\label{eq:shift}
  \beta(t)\Bigl(\pi_M - \frac{\eta}{\sigma}\Bigr)\frac{E_t}{X^*_t},\qquad t\in[0,T].
\end{equation}
Note that
\[
  \beta(T) = 1
  \qquad\text{and}\qquad
  \frac{\de}{\de t}\beta(t) = -\exp\bigl\{(\mu - r - \eta\theta)(T-t)\bigr\} < 0,
  \qquad t\in[0,T],
\]
from which we see that $\beta$ is strictly positive and strictly decreasing.
Thus, the sign of the shift is determined entirely by the constant
\[
  \pi_M - \frac{\eta}{\sigma} = \frac{1}{\gamma\sigma}\bigl(\theta - \gamma\eta\bigr).
\]
If $\gamma\eta < \theta$, the agent invests more into the risky asset compared to an investor without a random endowment.
For this case to occur, it is necessary that the market price of risk $\theta$ is strictly positive.
Moreover, either the agent's relative risk aversion $\gamma$ or the volatility $\eta$ of the random endowment have to be sufficiently small.
Otherwise, if $\gamma\eta > \theta$, the agent generally invests less of their wealth into the risky asset in the presence of a random endowment.

The magnitude of the shift in Equation~\eqref{eq:shift} is driven by two factors: time $t$ and the endowment-to-wealth ratio $E/X^*$.
If the endowment-to-wealth ratio $E_t/X^*_t$ is large---as would typically be the case for young agents---the shift can be quite significant as the dependence of the shift on the ratio is linear.
Conversely, the shift vanishes as the endowment-to-wealth tends to zero.
Finally, the impact of time $t$ on the magnitude of the shift is for the most part determined by the function $\beta$.
Recalling that $\beta$ is strictly positive and strictly decreasing, it follows that younger agents with larger investment horizons tend to deviate more significantly from $\pi_M$ than more mature agents with short investment horizons.

%%%
%%% Bibliography
%%%
\printbibliography

\end{document}